%% file: bare_conf.tex
\documentclass[conference]{IEEEtran}

\ifCLASSINFOpdf

\else

\fi

\usepackage{calc}
\usepackage{graphics,color,subfigure,wrapfig,hyperref}
\usepackage{color}
\usepackage{amsthm}
\usepackage{amsmath,amssymb}
\usepackage{amsfonts}
\usepackage{array}
\newcolumntype{P}[1]{>{\centering\arraybackslash}p{#1}}
\usepackage{multirow}
\usepackage{graphicx}
\usepackage{standalone}
\usepackage{booktabs}
\usepackage{makecell}

\usepackage{algorithm}
\usepackage{algpseudocode}

\usepackage{enumitem}
\theoremstyle{definition}

\newcommand{\calA}{\mathcal{A}}

\newcommand{\calP}{\mathcal{P}}
\newcommand{\calX}{\mathcal{X}}
\newcommand{\calG}{\mathcal{G}}
\newcommand{\calH}{\mathcal{H}}
\newcommand{\calI}{\mathcal{I}}

\newcommand{\bfT}{\mathbf{T}}
\newcommand{\bfG}{\mathbf{G}}
\newcommand{\bfg}{\mathbf{g}}
\newcommand{\bfu}{\mathbf{u}}

\newtheorem{lemma}{Lemma}

\newtheorem{definition}{Definition}

\newtheorem{example}{Example}

{\begin{list}{}%
         {\setlength{\leftmargin}{#1}}%
         \item[]%
}
{\end{list}}

\begin{document}

\title{Leveraging Coding Techniques for Speeding up Distributed Computing}

\author{\IEEEauthorblockN{Konstantinos Konstantinidis and Aditya Ramamoorthy}
\IEEEauthorblockA{Department of Electrical and Computer Engineering\\
Iowa State University\\
Ames, IA 50010\\
Emails: \{kostas, adityar\}@iastate.edu}
}

\maketitle

\begin{abstract}
Over the last decade, distributed computing frameworks such as
MapReduce, Hadoop and Spark have become ubiquitous. Large scale
clusters routinely process data that are on the orders of petabytes
or more. The sheer size of the data precludes the processing of the
data on a single computer. The philosophy in these methods is to
partition the overall job into smaller tasks that are executed on
different servers; this is called the map phase. This is followed
by a data shuffling phase where appropriate data is exchanged
between the servers. The final so-called reduce phase, completes
the computation.

One potential approach, explored in prior work for reducing the
overall execution time is to operate on a natural tradeoff between
computation and communication. Specifically, the idea is to run
redundant copies of map tasks that are placed on judiciously chosen
servers. The shuffle phase exploits the location of the nodes and
utilizes coded transmission. The main drawback of this approach is
that it requires the original job to be split into a number of map
tasks that grows exponentially in the system parameters. This is
problematic, as we demonstrate that splitting jobs too finely can
in fact adversely affect the overall execution time.

In this work we show that one can simultaneously obtain low
communication loads while ensuring that jobs do not need to be
split too finely. Our approach uncovers a deep relationship between
this problem and a class of combinatorial structures called
resolvable designs. Appropriate interpretation of resolvable
designs can allow for the development of coded distributed
computing schemes where the splitting levels are exponentially
lower than prior work. We present experimental results obtained on
Amazon EC2 clusters for a widely known distributed algorithm,
namely TeraSort. We obtain over 4.69$\times$ improvement in speedup over the baseline approach  and more than 2.6$\times$ over current state of the art.

\end{abstract}

%
\IEEEpeerreviewmaketitle

\section{Introduction}
\label{sec:intro}

In recent years, there has been a surge in the usage of various cluster computing frameworks such as MapReduce \cite{DeanG08}, Hadoop \cite{HADOOP} and Spark \cite{apache_spark}. The era of bigdata analytics \cite{Cuzzocrea:2011:AOL:2064676.2064695} whereby a large amount of data needs to be processed in a fast manner has fueled this growth. In these applications, datasets are often so large that they cannot be housed in the memory and/or the disk of any one computer. Thus, the data is typically distributed across a number of nodes. Each node performs its own local computation, following which there is a phase where the nodes communicate amongst themselves. At this point the nodes perform the final computation. Henceforth, we refer to this as the MapReduce framework.

The MapReduce framework has proven to be quite versatile and large scale clusters in industry and academia routinely process terabytes of data using this approach. It is important to note that the framework intertwines computation and communication. Specifically, multiple servers allow for parallel computation; yet data needs to be exchanged between them to complete the processing of the job. It is well-recognized that the data shuffling phase that occurs between the map and the reduce steps, results in a significant amount of data movement. In fact, the data shuffling phase limits the performance of several applications \cite{GuoRZ13}. Reference \cite{Chowdhury_etal11} suggests that for distributed machine learning algorithms\\

{\it ``... these transfers can have a significant impact on job performance accounting for more than $50$\% of the job completion times ..."}\\

There have been several papers on the topic of reducing the impact of the shuffle phase on the overall execution of a MapReduce job \cite{Chowdhury_etal11,Ahmad_2014,Cao_2016,Wang_2013,LiMAAllerton15,LiMA16,LiSMA17}.
Our work in this paper examines this problem within the computation vs. communication tradeoff. This tradeoff was first explored in the work of \cite{LiMAAllerton15,LiMA16,LiSMA17}. The ideas in these works have their origins in the problem of coded caching \cite{MAMAUNFLC}. Their work considers a model of MapReduce and defines appropriate notions of computation and communication loads. The key finding of their work is that the judicious usage of coded transmissions in the shuffle phase can fundamentally reduce the induced communication load.

The approach of \cite{LiMAAllerton15,LiMA16,LiSMA17} crucially relies on subdividing the file on which a given job is run into a large number of subfiles. In this work, our starting point is the realization that in practice this so-called subpacketization can rather adversely affect the overall execution time. Accordingly, we propose alternate mechanisms that allow us to seamlessly tradeoff computation vs. communication, but with acceptable levels of subpacketization. Our mechanisms are based on a natural link between error correcting codes, combinatorial objects known as resolvable designs and MapReduce-like protocols (related links in a different context were explored in \cite{TangR17, TangR17_preprint}). We present exhaustive experimental comparisons with prior work that demonstrate the efficacy of our method.

We note here that in recent years there have been several papers that have utilized the ideas from coding theory in various distributed storage, cloud computing and distributed computing tasks. While the original motivation for codes was protection against channel errors it turns out that interpreting them in the correct way can often result in interesting conclusions in different domains. Our work fits squarely within this overall picture. We use error correcting codes as a combinatorial object in our work to specify a MapReduce-like distributed computing protocol, which offers a powerful computation vs. communication tradeoff.

This paper is organized as follows. Section \ref{sec:formulation} discusses our problem formulation. The specification of our protocol and its corresponding analysis appear in Section \ref{sec:scheme}. We have implemented our scheme on Amazon EC2 clusters; details of this implementation can be found in Section \ref{sec:implementation}. Section \ref{sec:results} discusses our experimental results and Section \ref{sec:conclusions} concludes the paper.
\section{Problem Formulation}
\label{sec:formulation}
We now discuss the problem formulation more formally. Our presentation here is based closely on \cite{LiMA16}. There are $N$ input subfiles that correspond to disjoint parts of the entire file to be processed. There are $Q$ arbitrary output functions that need to be computed across these $N$ subfiles. There are a total of $K$ servers. For instance, in a word counting example, the subfiles could be the chapters of a book and the output functions are the word counts of a specific set of words. The subfiles will be denoted by $w_1, \dots, w_N$ and the output functions $\phi_j, j = 1, \dots, Q$. Each function $\phi_j$ depends on all the subfiles $w_1, \dots, w_N$.
We assume that the $j$-th function can be computed by a map phase followed by a reduce phase, i.e.,
\begin{align*}
\phi_j(w_1, \dots, w_N) = h_j(g_{j,1}(w_1), \dots, g_{j,N}(w_N)).
\end{align*}
Here, $\bfg_n = (g_{1,n}, \dots, g_{Q,n})$ ``maps" the subfile $w_n$ into $Q$ intermediate values $\nu_{j,n}, j = 1, \dots, Q$ each of which is assumed to be of size $B$ bits. The function $h_j$ maps the intermediate values $\nu_{j,n}$ on all subfiles into a reduced value $h_j(g_{j,1}(w_1), \dots, g_{j,N}(w_N))$.

\begin{example}
\label{eg:example_1}
Suppose that we consider the problem of counting $Q=4$ words of a collection $\calA = \{\text{and, if, when, the}\}$ in a book consisting of $N=4$ chapters in a cluster with $K=4$ servers. In this case the subfiles $w_1, \dots, w_4$ are the chapters and $\phi_i, i = 1, \dots, Q$ correspond to the counts of the words in $\calA$ in the entire book, e.g., $\phi_1(w_1, \dots, w_4)$ would be the number of occurrences of the word ``and" in the book. Suppose that we define $\bfg_n$ to be the function that returns the counts of all the words in $\calA$ in chapter $w_n$. Let us assume that the $i$-th slave is assigned subfile $w_i$ for all values of $i$. In this case it is evident that in the map phase, server $i$ computes $\bfg_i$ on its assigned subfile $w_i$ for $i = 1, \dots, 4$. In the reduce phase, each server is given the responsibility of finding the overall count in the book of one specific word, e.g., suppose that server 1 reduces the word ``and". In this case, it is evident that $\phi_1(w_1, \dots, w_N)$ can be computed as
\begin{align*}
\phi_1(w_1, \dots, w_N) = h_1(g_{1,1}(w_1), \dots, g_{1,N}(w_N)).
\end{align*}
In particular, the function $h_1$ simply corresponds to the sum of the counts of ``and" on the individual chapters.
\end{example}

We note here that the decomposition of a given job into map and reduce phases is not unique, i.e., these are design choices.

As noted in Section \ref{sec:intro}, there are several MapReduce jobs where the shuffle phase is rather time-intensive and can actually dominate the overall job execution time. Thus, it makes sense to operate on a tradeoff between communication and computation, i.e., one could choose to increase the computation load of the system by processing the same subfile at $r > 1$ servers. This would in turn reduce the number of intermediate values it needs in the reduce phase. Moreover, doing this systematically opens up possibilities for exploiting coded transmissions that can further reduce the communication load.  For a given network throughput rate, a lower communication load translates into lesser time taken in the shuffle phase. Thus, depending on the job characteristics it is plausible that the overall execution time of a given job can be reduced, i.e., the potential increase in the map phase execution time may be offset by the reduced shuffle phase time. Another important aspect here is the choice of the number of subfiles $N$. Note that this is also a design choice, e.g., in word counting, the subfiles can be chapters or pages of the book. As we will see (see Section \ref{sec:results}), there are performance issues if the number of subfiles is too large. This forms the main motivation of our work. In particular, our proposed scheme leverages much of the gains of a higher computation load, while operating with a much lower number of subfiles.

For the remainder of the paper, we refer to $r$ as the computation load. 
\begin{definition}
The communication load $L \in [0,1]$ of a certain scheme is defined as the ratio of the total number of bits transmitted in the data shuffling phase to $QNB$.
\end{definition}
In Example \ref{eg:example_1}, at the end of the map phase, each node needs three values from the other nodes. Thus, the total number of bits transmitted would be $4 \times 3 \times B = 12B$. Thus, the communication load of the system will be $L = 12B/16B = 3/4$.

\begin{figure}[t]
\centering
\includegraphics[scale=0.18]{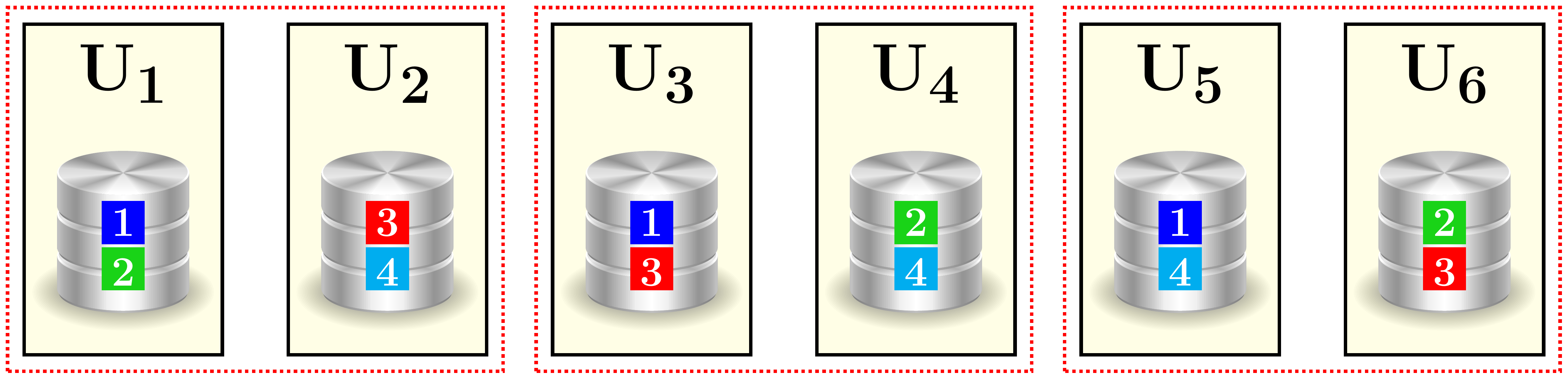}
\caption{Proposed placement scheme for $K=6$ servers and $N=4$ subfiles, denoted $\{1,2,3,4\}$ and represented by colored squares, each assigned to some server. The red dotted boxes show the partition of the files into parallel classes.}
\label{fig:resolv_eg}
\vspace{-0.2in}
\end{figure}


We now present an example which demonstrates that increasing $r$ can translate into lower communication loads.
\begin{example}
\label{eg:N4_example}
Consider a system with $K=6$ servers, a computation load of $r=3$ and $Q=6$ functions (e.g., word counts) that need to be computed. In our approach we would subdivide the original job into $N=4$ subfiles (more generally any multiple of $4$ can be used) that will be assigned to the servers as demonstrated in Fig. \ref{fig:resolv_eg}. At the end of the map step, each server would have computed the $Q$ functions on its assigned map tasks. Suppose that the $i$-th server is responsible for reducing the $i$-th function. This would imply, for example, that server $U_1$ needs the first function's evaluation on subfiles $w_3$ and $w_4$.

The key idea of our approach is for each server to transmit a packet that is simultaneously useful to multiple servers. For example, let us consider the group of servers $G_1=\{U_1,U_3,U_6\}$ that were assigned subfiles $\{w_1, w_2\}, \{w_1,w_3\}$ and $\{w_2, w_3\}$, respectively. Then it is evident that at the end of the map phase, server $U_1$ wants $\nu_{1,3}$, server $U_3$ wants $\nu_{3,2}$ and server $U_6$ wants $\nu_{6,1}$. We assume that $\nu_{j,n}$ can be encapsulated into a packet with size $B$ bits, denoted by $p(\nu_{j,n})$. Furthermore, assume that this packet can be subdivided into two parts $p(\nu_{j,n})[1]$ and $p(\nu_{j,n})[2]$ (with size $B/2$ bits)\footnote{Strictly speaking, such an assumption will not hold, owing to the bits allocated to header information etc. Nevertheless, we make this assumption to keep the discussion simple.}.


Now consider the set of transmissions specified in Table \ref{table:groups_eg1}. Note that server $U_1$ contains subfiles $w_1$ and $w_2$ and can therefore compute all $Q$ functions associated with them. Thus, it can transmit $p(\nu_{3,2})[1] \oplus p(\nu_{6,1})[2]$ as specified in row 1 of the table. Furthermore, it can be observed that this transmission is {\it simultaneously} useful to both servers $U_3$ and $U_6$. In particular, server $U_3$ already knows $p(\nu_{6,1})[2]$ and can therefore decode $p(\nu_{3,2})[1]$ which it wants. Likewise, server $U_6$ already knows $p(\nu_{3,2})[1]$ and can decode $p(\nu_{6,1})[2]$ that it wants. In a similar manner, it can be verified that each of the transmissions in Table \ref{table:groups_eg1} benefit two servers of the corresponding group. It turns out that the process of picking the servers to consider together can be made systematic; in addition to server group $G_1$ that we just considered, we can pick three others: $G_2=\{U_1,U_4,U_5\}$, $G_3=\{U_2,U_3,U_5\}$ and $G_4=\{U_2,U_4,U_6\}$ which will result in all the servers obtaining their desired values.

The total number of bits transmitted in this case is therefore $4 \times 3 \times B/2 = 6B$; thus, the corresponding communication load is $\frac{6B}{QNB} = 0.25$. In contrast, uncoded transmission from the different nodes would have required a total of $2 \times 6 \times B = 12B$ bits to be transmitted, corresponding to a communication load of $0.5$. Thus, the proposed approach reduces the communication load in the shuffle phase by half.
\end{example}

\begin{table}[t]
\centering
\caption{Coded transmissions within all groups of Example \ref{eg:N4_example}}
\begin{tabular}{ |c|c|c| }
\hline
Group&Server&Transmission\\
\hline
\multirow{3}{*}{$G_1$}
&$U_1$&$p(v_{3,2})[1]\oplus p(v_{6,1})[2]$\\
&$U_3$&$p(v_{6,1})[1]\oplus p(v_{1,3})[2]$\\
&$U_6$&$p(v_{3,2})[2]\oplus p(v_{1,3})[1]$\\
\hline
\multirow{3}{*}{$G_2$}
&$U_1$&$p(v_{5,2})[1]\oplus p(v_{4,1})[1]$\\
&$U_4$&$p(v_{5,2})[2]\oplus p(v_{1,4})[1]$\\
&$U_5$&$p(v_{4,1})[2]\oplus p(v_{1,4})[2]$\\
\hline
\multirow{3}{*}{$G_3$}
&$U_2$&$p(v_{5,3})[1]\oplus p(v_{3,4})[1]$\\
&$U_3$&$p(v_{5,3})[2]\oplus p(v_{1,1})[1]$\\
&$U_5$&$p(v_{3,4})[2]\oplus p(v_{1,1})[2]$\\
\hline
\multirow{3}{*}{$G_4$}
&$U_2$&$p(v_{6,4})[1]\oplus p(v_{4,3})[1]$\\
&$U_4$&$p(v_{6,4})[2]\oplus p(v_{2,2})[1]$\\
&$U_6$&$p(v_{2,2})[2]\oplus p(v_{4,3})[2]$\\
\hline
\end{tabular}
\label{table:groups_eg1}
\vspace{-0.2in}
\end{table}

We note here that the original work of \cite{LiMA16} promises a communication load of $L_{coded}(r) = \frac{1}{r} (1 - \frac{r}{K})$ for $r \in \{1, \dots, K\}$. However, crucially this result assumes that
\begin{align*}
N &= \binom{K}{r} \eta_1, \text{~where $\eta_1$ is a positive integer.}
\end{align*}
It is evident that $N$ grows very rapidly for their scheme. In Section \ref{sec:results}, we demonstrate that in real-life experiments the idealized analysis does not hold and the execution time suffers as a result of this.

For instance, for Example \ref{eg:N4_example} above, their communication load would be $1/6$ (which is lower). However, their approach would require the original file to be split into $\binom{6}{3} = 20$ subfiles, i.e., their scheme only works when $N = 20$. In contrast, the scheme proposed in Example \ref{eg:N4_example} works with $N=4$ which is much smaller.

In this work, we present significant generalizations of the basic approach in Example \ref{eg:N4_example}. We present protocols for subpacketizing a given file into $N$ subfiles. These subfiles are then distributed over $K$ servers. The servers work together within the Map/Shuffle/Reduce framework to compute $Q$ different functions on the original file. The system operates at a computation load of $r$. In the upcoming Section \ref{sec:scheme}, we present the details of our approach.
Our proposed protocol has been implemented within a cloud based cluster set up within Amazon Web Services (AWS). In Section \ref{sec:results}, we demonstrate experimental results that demonstrate the overall savings in execution time that our technique enjoys over an uncoded system and over previous approaches. In this section, we also discuss several implementation related issues and discuss the effect of the nature of the underlying computation on the overall execution time.

\section{Improved Schemes For Coded Distributed Computation from Resolvable Designs}
\label{sec:scheme}

In this section, we present the details of our proposed approach.
\subsection{Primer on resolvable designs}
We begin with some notions from combinatorial design theory \cite{DRSCDCA}.
\theoremstyle{definition}
\begin{definition}
A \textit{design} is a pair ($X$,$\mathcal{A}$) consisting of
\begin{enumerate}
\item a set of elements (\textit{points}), $X$, and
\item a family $\mathcal{A}$ (i.e. multiset) of nonempty subsets of $X$ called \textit{blocks}, where each block has the same cardinality.
\end{enumerate}
\end{definition}
Thus, a design is simply a set system, where each set (or block) has the same cardinality. It turns out that examining designs that have specific structure is especially useful in the distributed computing context. In this work, we will make use of \textit{resolvable designs}, a special category of block designs.
\begin{definition}
A subset $\mathcal{P}\subset X$ in a design $(X,\mathcal{A})$ is said to be a \textit{parallel class} if for $X_i\in\mathcal{P}$ and  $X_j\in\mathcal{P}$ with $i\neq j$ we have $X_i\cap X_j=\emptyset$  and $\cup_{\{j:X_j\in P\}}X_j=X$. A partition of $\mathcal{A}$ into several parallel classes is called a resolution, and $(X,\mathcal{A})$ is said to be a resolvable design if $\mathcal{A}$ has at least one resolution.
\end{definition}
A simple example of a resolvable design is obtained by considering all $2$-subsets of $\{1, \dots, 4\}$.
\begin{example}
\label{eg:subsets of_four}
Let $X=\{1,2,3,4\}$ and  $\mathcal{A}=\{\{1,2\},\{3,4\},\{1,3\},\{2,4\},\{1,4\},\{2,4\}\}$. The $(X, \calA)$ forms a resolvable design with the following parallel classes.
\begin{IEEEeqnarray*}{rCl}
\mathcal{P}_1&=&\{\{1,2\},\{3,4\}\},\\
\mathcal{P}_2&=&\{\{1,3\},\{2,4\}\}, \text{~and}\\
\mathcal{P}_3&=&\{\{1,4\},\{2,3\}\}.
\end{IEEEeqnarray*}
We note here that Example \ref{eg:N4_example} used precisely this design, when specifying the subpacketization and the placement.
\end{example}
It turns out that there is a systematic procedure for constructing resolvable designs, where the starting point is an error correcting code \cite{lincostello}. We explain this procedure below.

Let $\mathbb{Z}_q$ denote the additive group of integers modulo $q$. The generator matrix of an $(k,k-1)$ single parity-check (SPC) code over $\mathbb{Z}_q$\footnote{We emphasize that this construction works even if $q$ is not a prime, i.e., $\mathbb{Z}_q$ is not a field.} is defined by
\begin{equation}
\mathbf{G}_{SPC}=
\begin{bmatrix}
& &\vline&1\\
&\huge \mathbf{I}_{k-1}&\vline&\vdots\\
&&\vline&1
\end{bmatrix}.
\label{eqn:spcgen}
\end{equation}
This code has $q^{k-1}$ codewords. The codewords are $\mathbf{c} = \mathbf{u} \cdot \mathbf{G}_{SPC}$ for each possible message vector $\mathbf{u}$. The code is systematic so that the first $k-1$ symbols of each codeword are the same as the symbols of the message vector. The $q^{k-1}$ codewords $\mathbf{c}_i$ computed in this manner are stacked into the columns of a matrix $\mathbf{T}$ of size $k\times q^{k-1}$.
\begin{equation}
\mathbf{T}=[{\mathbf{c}}_1^T,{\mathbf{c}}_2^T,\cdots,{\mathbf{c}}_{q^{k-1}}^T].
\label{eqn:T}
\end{equation}
The corresponding resolvable design is constructed as follows. Let $X_{SPC} = [q^{k-1}]$ (we use $[n]$ to denote the set $\{1,2,\dots,n\}$ throughout) represent the point set of the design.
We define the blocks as follows. For $0 \leq l \leq q-1$, let $B_{i,l}$ be a block defined as
$$B_{i,l}=\{j: \mathbf{T}_{i,j}=l\}.$$
The set of blocks $\mathcal{A}_{SPC}$ is given by the collection of all $B_{i,l}$ for $1 \leq i \leq k$ and $0 \leq l \leq q-1$ so that $|\mathcal{A}_{SPC}| = kq$. The following lemma (see Appendix for proof) shows that
this construction always yields a resolvable design. This lemma was proved in \cite{8007038} in a different context.
\theoremstyle{lemma}
\begin{lemma}
\label{lemma:cons_yields_resolv_design}
The above scheme always yields a resolvable design $(X_{SPC},\mathcal{A}_{SPC})$ with $X_{SPC}=[q^{k-1}]$, $|B_{i,l}|=q^{k-2}$ for all $1\leq i\leq k$ and $0\leq l\leq q-1$. The parallel classes are analytically described by $\mathcal{P}_i=\{B_{i,l}:0\leq l\leq q-1\}$, for $1\leq i\leq k$.
\end{lemma}

\begin{example}
\label{ex:classicspc}

The generator matrix of this $(3,2)$ SPC code over $\mathbb{Z}_2$ (binary) is given by
$$\mathbf{G}_{SPC}=
\begin{bmatrix}
1&0&1\\
0&1&1
\end{bmatrix}.
$$
The matrix $\mathbf{T}$ can be obtained as
$$\mathbf{T}=[{\mathbf{c}}_1^T,{\mathbf{c}}_2^T,{\mathbf{c}}_3^T,{\mathbf{c}}_4^T]=\begin{bmatrix}
0&0&1&1\\
0&1&0&1\\
0&1&1&0
\end{bmatrix}.
$$
It can be observed, e.g., that $B_{1,0} = \{1,2\}$ and $B_{1,1} = \{3,4\}$ so that they form a parallel class. In fact, this construction returns the resolvable design considered in Example \ref{eg:subsets of_four}.
\end{example}

\subsection{From resolvable designs to protocol specification}

The main idea in our work is to use an appropriate resolvable design to specify the number of subfiles, the map task assignments and the messages transmitted in the shuffle phase
for a given distributed computing job. We explain this correspondence next.

\begin{algorithm}[!t]
\caption{Proposed Protocol}
\label{Alg:protocol}
\begin{algorithmic}[1]
\State Input: File $\mathcal{W}$, $Q$ functions; number of servers $K = k \times q$. $K$ divides $Q$.
\State Use a $(k,k-1)$ SPC code to generate a design $(\mathcal{X},\mathcal{A})$
\State Split $\mathcal{W}$ into $q^{k-1}$ disjoint subfiles, $w_{1},\dots,w_{q^{k-1}}$.
\State Assign subfiles to servers such that server $B_{i,j}$ is assigned subfile $w_\ell$ if $\ell \in B_{i,j}$.
\State Choose an equal-size partition of $[Q]$ to obtain the sets $\phi^{B_{i,j}}$ for $i = 1, \dots, k$ and $j = 0, \dots, q-1$.
\State Execute the Map phase on each of the servers.
\State Choose all possible sets $\{B_{1,j_1}, B_{2, j_2}, \dots, B_{k,j_k}\}$ where $j_\ell \in \{0, \dots, q-1\}$, such that $\cap_{\ell=1}^k B_{\ell, j_\ell} = \emptyset$ and store them in a collection $\calG$.
\For{$\gamma \in [Q/K]$}
\For{each group $G=\{B_{1,j_1}, B_{2, j_2}, \dots, B_{k,j_k}\}\in \calG$}
\State Determine $\Delta^G_\ell = \nu_{\phi^{B_{i,j}}[\gamma], \cap_{k\neq \ell} B_{k,j_k}}$ for $\ell = 1, \dots, k$.
\State Split packet $p(\Delta^G_\ell)$ into $k-1$ parts at each server where it is available.
\State Execute BIPARTITE-STEP$(\{\Delta^G_\ell\}_{\ell = 1}^k, G)$
\State For each $\ell \in [k]$, server $B_{\ell, j_\ell}$ transmits $$\bigoplus_{m \neq \ell} p(\Delta^{G}_m)[\text{label}(\Delta^{G}_m,B_{\ell, j_\ell})]$$
\EndFor
\EndFor
\State Execute Reduce phase on each of the servers.
\end{algorithmic}
\end{algorithm}

Consider a file $\mathcal{W}$ on which $Q$ functions need to be computed and suppose that we have access to $K$ servers; we assume that $Q$ is a multiple of $K$. In Algorithm \ref{Alg:protocol}, we present the steps that specify the entire protocol.
The protocol can be understood as follows. We choose an integer $q$ such that $q$ divides $K$, i.e., $K = k \times q$. Next, we form a $(k,k-1)$ SPC code and the corresponding resolvable design using the aforementioned procedure. The point set $\calX = [q^{k-1}]$ and the block set $\calA$ will be such that $|\calA| = kq$. The blocks of $\calA$ will be indexed as $B_{i,j}, i = 1, \dots, k$ and $j = 0, 1, \dots, q-1$.

We associate the point set $\calX$ with the subfiles, i.e., $N = |\calX| = q^{k-1}$ and the block set $\calA$ with the servers. The map task assignment follows the natural incidence between the points and the blocks, i.e., server $B_{i,j}$ is responsible for executing the map tasks on the set of subfiles $\text{Map}[B_{i,j}] = \{w_\ell ~|~ \ell \in B_{i,j}\}$. Thus, at the end of the map phase, server $B_{i,j}$ has computed the $Q$ intermediate value on the subfiles in $\text{Map}[B_{i,j}]$.

Recall that we assume that $K$ divides $Q$. Thus, each server is responsible for reducing $Q/K$ functions. We let $\phi^{B_{i,j}} \subset [Q]$ represent the set of functions assigned for reduction to server $B_{i,j}$. The sets $\phi^{B_{i,j}}$ form a partition of $[Q]$. For ease of notation, we let $\phi^{B_{i,j}}[\ell]$ represent the $\ell$-th function in the set $\phi^{B_{i,j}}$; $\ell$ ranges from $1$ to $Q/K$.

Following the map phase, in the shuffle phase, each server $B_{i,j}$ needs intermediate values from other servers so that it has enough information to reduce the functions in $\phi^{B_{i,j}}$. In this step we transmit coded packets that are simultaneously useful to multiple servers. Towards this end we form a collection of user groups by choosing one block from each parallel class according to the rule in Step 7 of the protocol, i.e., we choose servers $B_{1,j_1}, B_{2, j_2}, \dots, B_{k,j_k}$ such that $\cap_{\ell=1}^k B_{\ell, j_\ell} = \emptyset$. For a given user group $G$ (of size $k$) we can show that each server in $G$ can transmit a useful message to $k-1$ other servers. Furthermore, we can show that considering all possible user groups allows the shuffle phase to achieve its objective, i.e., at the end of the shuffle phase, all servers have enough information to execute the reduce phase. An instance of the shuffle phase equations was discussed in Example \ref{eg:N4_example} (see Table \ref{table:groups_eg1}).

\begin{figure}[t]
\centering
\includegraphics[scale=1]{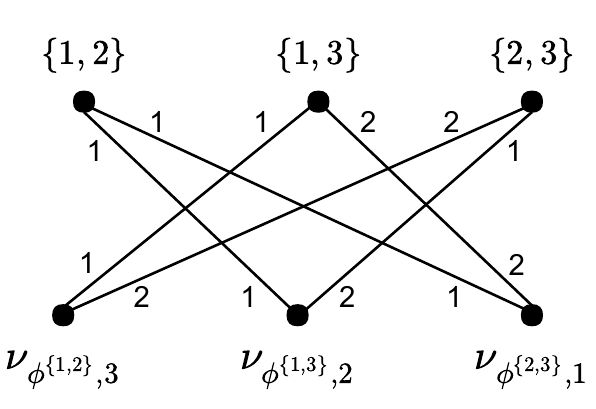}
\caption{Bipartite graph corresponding to the user group $G_1$ in Example \ref{eg:N4_example}.}
\vspace{-0.2in}
\label{fig:bipartite}
\end{figure}

\subsection{Proof of correctness and communication load analysis}

We now prove that the proposed protocol allows each server $B_{i,j}$ to recover enough information at the end of the shuffle phase. As the protocol is symmetric with respect to blocks, we equivalently show that server $B_{1,j_1}$ is satisfied. Note that $|B_{1,j_1}| = q^{k-2}$. For the purposes of our arguments below, we assume that $Q = K$. The case when $Q$ is an integer multiple of $K$ is quite similar. In this case, with some abuse of notation, since $\phi^{B_{1,j_1}}$ is a singleton set, we use $\phi^{B_{1,j_1}}$ to actually represent the function index itself. It is therefore clear that $B_{1,j_1}$ needs the intermediate values $\nu_{\phi^{B_{1,j_1}},n}$ for $n \in [q^{k-1}] \setminus B_{1,j_1}$.

Now consider the construction of the user groups. Let $G$ be a user group where $B_{1,j_1}$ is chosen from $\calP_1$, i.e., $G = \{B_{1,j_1},B_{2,j_2}, \dots, B_{k,j_k}\}$. In Lemma \ref{claim:intersect} in the Appendix, we show that the intersection of {\it any} $k-1$ blocks from $k-1$ distinct parallel classes is always of size $1$. Furthermore, note that the intersection of all the blocks in $G$ is empty. We will show that in the transmissions corresponding to user group $G$, the server $B_{\ell,j_\ell}$ can obtain intermediate value $\Delta^{G}_\ell = \nu_{\phi^{B_{\ell,j_\ell}}, \cap_{k \neq \ell} B_{k,j_k}}$ which it wants. To do this we consider the construction of the bipartite graph shown in Algorithm \ref{Alg:BIPARTITE-STEP}. An example of this bipartite graph corresponding to user group $G_1$ in Example \ref{eg:N4_example} appears in Fig. \ref{fig:bipartite}. The bottom nodes in the graph represent the useful intermediate values received by each node in $G$. The top nodes represent the servers in $G$.

An edge in the graph represents the fact that a server is capable of generating the intermediate value on the bottom. We arbitrarily label the incoming edges into a bottom node from $1$ to $k-1$. Following this labeling, we can determine the exact transmission corresponding to user group $G$. In particular, according to the protocol, server $B_{\ell, j_\ell}$ transmits
$$ \bigoplus_{m \neq \ell} p(\Delta^{G}_m)[\text{label}(\Delta^{G}_m,B_{\ell, j_\ell})].$$
Thus, the packet corresponding to each useful intermediate value is subdivided into $k-1$ equal parts. Then, it is evident that the equations transmitted above, allow each $B_{\ell,j_\ell}$ to recover $\Delta^{U}_\ell$ for $\ell = 1, \dots, k$. The total number of bits transmitted in processing user group $G$ is therefore $B \times \frac{k}{k-1}$.


\begin{algorithm}[!t]
\caption{BIPARTITE-STEP Algorithm}\label{Alg:BIPARTITE-STEP}
\begin{algorithmic}[1]
\State Input: Block $G = \{B_{1,j_1}, \dots, B_{k,j_k}\}$, and $\{\Delta^G_\ell\}_{\ell =1}^k$.
\State Construct a bipartite graph $\calH$ with $k$ nodes on the top corresponding to $G$ and $k$ nodes on the bottom corresponding to
$\{\Delta^G_\ell\}_{\ell =1}^k$.
\State For $\ell \in [k]$, connect $\Delta^G_\ell$ to the nodes in $G \setminus \{B_{\ell, j_\ell}\}$.
\For{$\ell \in [k]$}
\State In $\calH$, label the edges incident on $\Delta^G_\ell$, arbitrarily with distinct labels from $1, \dots, k-1$.
\EndFor
\end{algorithmic}
\end{algorithm}

We conclude the proof by observing that a given block, e.g., $B_{\ell,j_\ell}$ participates in $q^{k-2} (q-1) = q^{k-1} - q^{k-2}$ user groups each of which allow it to obtain distinct intermediate values. This can be seen as follows. Suppose for instance, that
$$ \cap_{k \neq \ell} B_{k,j_k} = \cap_{k \neq \ell} B_{k,j_k'}$$
where $j_m \neq j_m'$ for at least one value of $m \in [k] \setminus \{\ell\}$. In this case, we note that the equality above implies that
$\cap_{k \neq \ell} B_{k,j_k} \bigcap \cap_{k \neq \ell} B_{k,j_k'} \neq \emptyset$. This is a contradiction, because $B_{m,j_m} \cap B_{m,j_m'} = \emptyset$ as they are two blocks belonging to the same parallel class.

Therefore, since $B_{\ell,j_\ell}$  is missing exactly $q^{k-1} - q^{k-2}$ intermediate values, it follows that at the end of the shuffle phase it is satisfied. By symmetry, therefore all users are satisfied.

Next, we present the analysis of the communication load of our algorithm. In the uncoded case, each server needs $QN/K$ intermediate values $\nu_{j,n}$'s to execute its reduce phase. Note that each server already has $rN/K \times Q/K$ of them. Thus, in the shuffle phase the communication load is given by
\begin{align*}
L_{uncoded} &= \frac{K(QN/K - rQN/K^2)B}{QNB}\\
&= 1 - \frac{r}{K}.
\end{align*}
On the other hand, for our scheme, the number of bits transmitted in shuffle phase is given by
$$ q^{k-1}(q-1) \cdot B\frac{k}{k-1} \cdot \frac{Q}{K}.$$
Thus, the communication load is given by
\begin{align*}
L_{prop} &= \frac{q^{k-1}(q-1) \cdot B\frac{k}{k-1} \cdot \frac{Q}{K}}{QNB}\\
&= \frac{1}{k-1} \bigg{(} 1 - \frac{k}{K} \bigg{)},
\end{align*}
where the second equality above is obtained by using the fact that $N = q^{k-1}$ and $K = kq$.

Next, note that for our proposed scheme the computation load is $k$, i.e., $r=k$. Thus, we reduce the overall communication load by a factor of $\frac{1}{r-1}$ with respect to an uncoded system.
In contrast, the approach of \cite{LiMA16}, reduces the communication load by a factor of $\frac{1}{r}$. However, this comes at the expense of a large $N$ as discussed previously.

\section{Details of Implementation}
\label{sec:implementation}
We have implemented TeraSort on Amazon EC2 clusters  using our proposed approach. The implementation was performed in C++ using the Open MPI library for communication among the processes of the master and the servers. Our code builds on \cite{USCTS}
and comparisons with the uncoded case and the approach of \cite{LiSMA17} have been made. In this section, we describe the configuration of these clusters as well as implementation details.

TeraSort is a popular benchmark that measures the time to sort a big amount of randomly generated data on a given cluster of computers. For example, Apache Hadoop \cite{HADOOP}, which is a popular distributed computing platform provides a standard software library that contains its implementation. TeraSort has a long history which goes back to 2008 when Yahoo! managed to set a record sorting a 1TB of data in 209 seconds using a Hadoop cluster of 910 node \cite{TERASORT2008}. The data set in TeraSort is such that each line of the file is a key-value (KV) pair typically consisting of an integer key and a arbitrary string value. The sorting is done based on the key. It is not too hard to see that this KV formulation can be put in on-to-one correspondence with the formulation in terms of map and reduce functions ({\it cf.} Section \ref{sec:formulation}). As the KV terminology is quite well-established for TeraSort, we discuss it in these terms in the subsequent discussion.




\subsection{Amazon EC2 cluster configuration}
We used Amazon EC2 instances among which one served as a master and the rest of them as slaves (servers). The specifications of these machines are given in Table \ref{table:ec2cluster}. After placing the subfiles to the carefully chosen servers we also impose a limit of 100Mbps for both incoming and outgoing traffic of all machines\footnote{In order to manipulate the traffic control settings, we use the Linux $\texttt{tc}$ command}. This serves the purpose of alleviating bursty TCP transmissions.
\subsection{Data set description}
For the TeraSort experiments we generated 12GB of total data. Each row of the file holds a 10-byte key (unsigned integer) and its corresponding 90-byte value (arbitrary string). The TeraGen utility of Hadoop distribution was used to randomly generate this data. The KV pairs are lexicographically sorted with respect to the ASCII code of their keys where the leftmost and the rightmost byte are the most and the least significant byte, respectively.


\begin{table}[!t]
\centering
\caption{Specifications of the Amazon EC2 machines of the cluster}
\begin{tabular}{|>{\centering}m{1cm}||>{\centering}m{1cm}|>{\centering}m{1cm}|>{\centering}m{1cm}|>{\centering\arraybackslash}m{1.5cm}|}
\hline
Machine role&Instance type&Virtual CPUs&Memory&Storage\\
\hline
Master&r3.large&2&15.25GB&32GB SSD\\
\hline
Server&m3.large&2&7.5GB&32GB SSD\\
\hline
\end{tabular}
\label{table:ec2cluster}
\vspace{-0.2in}
\end{table}

\subsection{Platform and code implementation description}

The source code is in C++ and we used the Open MPI library, version 1.10.2, for communication among the processes of the master and the servers. Our source code repository is available at \cite{KKCT}.

The master node is responsible for placing the subfiles in the local drives of the servers and partitioning the domain of the keys. This partitioning step is equivalent to deciding the reducer responsibilities for each slave. It also initiates the MPI program to all servers. From this point onwards, the master will only take time measurements from the servers at the end of each MapReduce phase. 

The partitioning process can be explained as follows.
The master initially takes $s \geq K-1$ samples from the data set and sorts them. The $i$-th smallest sample will be $a_i$. The sorted samples are
$$\mathbf{a} = \left[\begin{IEEEeqnarraybox*}[][c]{,c/c/c/c,}a_1 & a_2 & \dots & a_s\end{IEEEeqnarraybox*}\right].$$
Assume that there are $Q=K$ functions to be reduced. Subsequently, we will pick $K-1$ equally spaced samples from $\mathbf{a}$, i.e., we pick
$$B=\Big\{a_i:i=\Big\lceil\frac{js}{K}\Big\rceil, j\in[K-1]\Big\},$$ and collect them into the vector
$$\mathbf{b}=\left[\begin{IEEEeqnarraybox*}[][c]{,c/c/c/c,}b_1 & b_2 & \dots & b_{K-1}\end{IEEEeqnarraybox*}\right].$$
The domain of the keys is partitioned into $K$ non-overlapping domains based on the keys in $\mathbf{b}$ such that $i$-th partition can be described by
\[
C_i= \left\{
\begin{array}{ll}
      \{d\in D:d<b_1\} & i=1, \\
      \{d\in D:b_{i-1}\leq d<b_i\} & i=2,\dots,K-1,\\
      \{d\in D:d\geq b_{K-1}\} & i=K. \\
\end{array}
\right.
\]
Thus, server $i$ will reduce the keys of partition $C_i$. The above construction guarantees that, at the end of the algorithm, the output of $i$-th server has smaller keys than the output of the $(i+1)$-th server.

To speed up the partitioning, the partitioner builds a two level trie that quickly indexes into the list of sample keys based on the first two bytes of the key. The root of trie has 256 (equal to the total number of ASCII characters) children and each of them has another 256 children. The root is associated with the empty key and all descendants of a node have a common prefix of the key associated with that node. The trie structure is typically chosen due to its key lookup efficiency as well as the guarantees it provides against the occurrence of hash collisions.

The only I/O operations at the slaves are those of reading the KV pairs from the subfiles. The intermediate data resides only in the memory during the Encoding, Shuffling and Decoding phases so there is no I/O involved during these stages.


The overall sequence of steps in processing a given job are: CodeGen $\rightarrow$ Map $\rightarrow$ Pack/Encode $\rightarrow$ Shuffle $\rightarrow$ Unpack/Decode $\rightarrow$ Reduce. We explain these steps below.
\begin{itemize}[leftmargin=4mm]
\item \textit{Code generation}: All nodes (including the master) start by generating the resolvable design based on our choice of the parameters $q$ and $k$. Next, the data set is split into $N$ subfiles by the master and the appropriate subfiles are transmitted to each slave based on Step 4 of the protocol.
    %
    The master also broadcasts the partition boundaries i.e. the keys that describe the partitioning of the domain at this point. Following the CodeGen phase, each slave will hash all keys of its subfiles to their partitions based on the partition boundaries.
\item \textit{Map.} 
    For each subfile $w_a$ that server $B_{i,l}$ has in its block, it will compute $\{\nu_{1,a},\dots,\nu_{Q,a}\}$ during the Map phase.
\item \textit{Pack/Encode.} For the uncoded implementation, we use the Pack operation. The Pack stage stores all intermediate values that will be sent to the same reducer in a continuous memory array so that a single TCP connection for each sender/receiver pair suffices (which may transmit multiple KV pairs) when $\texttt{MPI\_Send}$ is
called\footnote{In the Shuffling phase of the uncoded case, each server unicasts data to a single receiver at any particular time, which is exactly the purpose of $\texttt{MPI\_Send}$ call.}.

    In the coded implementation encoded packets are created from the mapped data as described in Algorithm \ref{Alg:protocol}. 
\item \textit{Shuffle.} For each shuffling group $G$ a server belongs to, it will broadcast an appropriate encoded packet to the rest of the group.
\item \textit{Unpack/Decode.} In the uncoded implementation we use the Unpack operation which simply deserializes the received data to a list of KV pairs.
In the coded implementation the intermediate values are decoded locally on each server from the received data. 
\item \textit{Reduce.} The reduce function is applied on the unpacked/decoded data.
\end{itemize}

%

\begin{table}[t]
\centering
\caption{Measurements for sorting 12GB data on 16 server nodes without coding}
\begin{tabular}{ |P{0.5cm}|P{0.5cm}|P{0.8cm}|P{0.8cm}|P{0.8cm}|P{1.25cm}|P{0.8cm}| }
 \hline
 Map&Pack&Shuffle&Unpack&Reduce&Total Time&Rate\\
 (sec.)&(sec.)&(sec.)&(sec.)&(sec.)&(sec.)&(Mbps.)\\
 \hline
 3.36&2.55&999.84&2.19&12.45&1020.39&100.80\\
 \hline
\end{tabular}
\label{table:uncoded}
\vspace{-0.2in}
\end{table}

We should emphasize that the groups formed during the CodeGen phase are stored in a structure according in the same order in the memory of all server machines. This facilitates a consistent shuffling phase where at any particular time it is clear which group is communicating, which nodes are the transmitters and which are the receivers. For the coded implementation there are two approaches we have tested to facilitate shuffling.
\begin{itemize}[leftmargin=4mm]
\item Using $\texttt{MPI\_Bcast}$ call: For this approach each server iterates through all shuffling groups that it belongs to and broadcasts a coded packet to the rest of the servers in the group. Thus, there is one sender and multiple receivers.
\item Using $\texttt{MPI\_Allgatherv}$ call: In this case, each server iterates through all groups of the scheme and checks for participation. If it belongs to the current group, it waits for the rest of the servers of the group to call this function. The exchange of all coded packets within the group ensues. This is a many-to-many communication pattern. Specifically, all nodes send and receive encoded data.
\end{itemize}

\section{Results and Discussion}
\label{sec:results}
We now present and discuss our experimental results. The results are based on the experiments we ran on 16 Amazon EC2 machines (servers) and the data set size is 12GB.
Table \ref{table:uncoded} corresponds to a uncoded TeraSort with $r=1$. It shows that the shuffle phase which takes $999.84$ seconds, dominates the overall execution time by far.

\begin{table*}[!t]
\centering
\caption{MapReduce time for sorting 12GB data on 16 server nodes including the memory allocation cost}
\label{table:tests}
\begin{tabular}{ |P{2cm}||P{1.1cm}|P{0.7cm}|P{1cm}|P{1.1cm}|P{1.1cm}|P{0.9cm}|P{1cm}|P{1cm}|P{0.9cm}|P{0.9cm}|P{0.9cm}| }
\hline
&\multirow{3}{*}{CodeGen}&\multirow{3}{*}{Map}&\multirow{3}{*}{\makecell{Pack/\\ Encode}}&\multirow{3}{*}{Shuffle}&\multirow{3}{*}{\makecell{Unpack/\\ Decode}}&\multirow{3}{*}{Reduce}&\multicolumn{2}{c|}{Total Time}&\multicolumn{2}{c|}{Speedup}&\multirow{3}{*}{Rate}\\
&\multirow{3}{*}{(sec.)}&\multirow{3}{*}{(sec.)}&\multirow{4}{*}{(sec.)}&\multirow{3}{*}{(sec.)}&\multirow{4}{*}{(sec.)}&\multirow{3}{*}{(sec.)}&\multicolumn{2}{c|}{(sec.)}&\multicolumn{2}{c|}{}&\multirow{3}{*}{(Mbps.)}\\
\cline{8-9}\cline{10-11}
&&&&&&&\multirow{2}{*}{w/MA}&w/out MA&\multirow{2}{*}{w/MA}&w/out MA&\\
\hline
Uncoded: &-&5.71&11.75&1105.64&4.46&12.88&1140.44&1126.68&&&100.83\\
Prior: $r=3$&5.82&17.94&229.80&455.05&6.23&14.54&729.38&496.76&$1.56\times$&$2.27\times$&64.79\\
Prior: $r=5$&26.78&29.99&1000.15&297.28&8.16&16.47&1378.83&490.28&$0.83\times$&$2.30\times$&61.04\\
Prior: $r=8$&38.41&51.03&1128.16&-&-&-&-&-&-&-&-\\
Proposed: $r=4$&0.64&25.91&9.93&307.15&6.91&17.29&367.83&352.91&$3.10\times$&$3.19\times$&88.47\\
Proposed: $r=8$&0.61&62.46&26.22&127.43&8.38&17.85&242.95&204.58&$4.69\times$&$5.51\times$&62.68\\
\hline
\end{tabular}
\end{table*}

\begin{table*}[!t]
\centering
\caption{Memory allocation cost for sorting 12GB data on 16 server nodes}
\label{table:ma}
\begin{tabular}{ |P{2cm}||P{0.7cm}|P{1.4cm}|P{1.8cm}|P{1.3cm}| }
\hline
&Map&Pack/Encode&Unpack/Decode&Total Time\\
&(sec.)&(sec.)&(sec.)&(sec.)\\
\hline
Uncoded: &2.32&9.17&2.27&13.76\\
Prior: $r=3$&6.87&223.16&2.59&232.62\\
Prior: $r=5$&11.29&874.14&3.12&888.55\\
Prior: $r=8$&18.03&968.48&-&-\\
Proposed: $r=4$&9.91&1.85&3.16&14.92\\
Proposed: $r=8$&22.01&13.17&3.19&38.37\\
\hline
\end{tabular}
\vspace{-0.2in}
\end{table*}

\subsection{Experimental Results}
Tables \ref{table:tests} and \ref{table:ma} contain the results of TeraSort using our approach and comparisons with the approach of \cite{LiSMA17}. There are approximately $130\times10^6$ KV pairs to be sorted. Table \ref{table:tests} presents the time each phase needed to completed including the memory allocation time. We have included a column for the total time that omits the memory allocation cost and another column that shows the speedup in that case. Table \ref{table:ma} presents only the memory allocation time of certain MapReduce phases for the same experiments. The need to take into account the memory allocation cost comes from the fact that for data sets at this scale, dynamic memory allocation on the heap (using the C++ $\texttt{new}$ operator) has a non-negligible impact on the total time. We note here that the results in \cite{LiSMA17} are generated using the code in \cite{USCTS} which explicitly ignores the memory allocation time ({\it cf.} communication with the first author of \cite{LiSMA17}). In our implementation (available at \cite{KKCT}), we measure the memory allocation time as well. 
We emphasize however, that the results in Table \ref{table:tests} indicate that our approach is consistently superior whether or not one takes into account the memory allocation time.

To understand the effect of choosing different values of $N$, we applied our algorithm with different values of $(k,q)=(r,q)$ pairs. The numbers reported at each line of the tables refer to a single experiment since we observed that after executing multiple times the behavior of the algorithm is consistent and the change is negligible compared to the total execution time. The network rate also remains the same across different runs.

We observe from Table \ref{table:tests} that if we account for memory allocation cost, our scheme achieves up to $4.69\times$ speedup compared to the uncoded TeraSort whereas if we ignore this code our schemes demonstrates an improvement of up to $5.51\times$. Moreover, the gain over the prior coded TeraSort scheme, if we compare the best time reported by each scheme, can go up to $4.69/1.56\approx3.01\times$ (when including memory allocation time) or $5.51/2.30\approx2.4\times$ (when excluding memory allocation time). We note here that the shuffle phase results corresponding to $r=8$ for prior work could not be obtained as their program crashed.

The following inferences can be drawn from Tables \ref{table:tests} and \ref{table:ma}.
\begin{itemize}[leftmargin=2mm]
\item Our CodeGen phase is quite efficient since the number of groups we need to generate and subsequently the number of intra-communicators we need to split the servers' communicator into, is quite smaller than that of the prior scheme. For example, let us look at the CodeGen time for $r=3$ of the prior scheme which is $t_1=5.82$. The corresponding number of groups is $g_1={{K}\choose{r+1}}={{16}\choose{4}}=1820$. For our scheme, that time is $t_2=0.64$ and the number of multicast groups is $g_2=q^{r-1}(q-1)=4^3\times3=192$. Now if we try to interpolate our code generation cost from $t_1$, based on our analysis, we would get:
$$t_2^{'}=\frac{g_2}{g_1}\times t_1=\frac{192}{1820}\times5.82\approx 0.61\approx t_2$$
\item The Map time mainly depends on the computation load $r$. We should recall at this point that $r$ is the number of times the whole data set is replicated and processed across the cluster. In that sense, we expect the Map cost of both coded schemes to be approximately $r$ times higher than that of the uncoded implementation of the algorithm. Indeed, if we look at our scheme for $r=4$ we see that $\frac{25.91}{5.71}\approx4.54$ is a good approximation to $r$. The overhead can be attributed partly to the fact that in the coded implementation each server has to open multiple input streams to read the subfiles from the disk while in the uncoded case each server opens only one stream. 
\item The encoding time of the coded schemes is not directly comparable to the packing of the uncoded approach. Nevertheless, we have a significant benefit over the prior scheme. For $r=8$, we obtain a speedup of $\frac{1128.16}{26.2}\approx43.06$.
This is explained by the fact that in the previous scheme each slave participates into much more groups and thus it needs to store more encoded data into its memory. Even the encoding time of the prior scheme for $r=3$ is quite expensive and takes $229.80$ seconds. 
\item The shuffle phase is where we can see the advantages of our implementations. For example, when $r=8$, our predicted load will be $1/14$, while the load of the uncoded $r=1$ scheme will be $15/16$. Thus, with the same transmission rate we expect our shuffle phase to be $13.125$ times faster. However, our obtained transmission rate is approximately $62.68$ Mbps. Thus, the overall gain is expected to be around $8.16$ times. In the actual measurements our gain is $\frac{1105.64}{127.43}\approx8.68$ which is quite close to the prediction.

    On the other hand, let us consider the prior scheme when $r=5$. In this case the load analysis predicts a gain of $6.82$ assuming that the transmission rates are the same and a gain of $4.127$ when accounting for the different rates. However, the actual gain is $\frac{1105.64}{297.28}\approx 3.72$ which is a little lower.

    Some of these discrepancies can be explained by the fact that the cost of multicasting a message from a node to $n$ receivers is not necessarily $n$ times cheaper than unicasting that message separately to each of the $n$ receivers. However, another factor that hurts the prior scheme is the large number of user groups that need to be considered in the shuffle phase, e.g, when $r=5$, the prior scheme needs to consider $\binom{K}{r+1} = 8008$ user groups, whereas our proposed scheme (for $r=8$), only needs to consider $q^{k-1}(q-1) = 128$ user groups. Thus, the overhead of setting the connections is much lesser in our protocol. Moreover, note that the lower layers of network protocols introduce additional headers in each packet. As the payloads in the prior scheme are smaller, the header overhead is also likely to affect them more.

\end{itemize}


\subsection{Discussion}
The major issue of the prior scheme \cite{LiSMA17} is the large value of $N$ that it needs. This translates into a large number ($\binom{K}{r+1}$) of user groups in the shuffling phase.
This number can be prohibitive for today's High Performance Computing (HPC) communication protocols like the Message Passing Interface (MPI). This is because all MPI communication is associated with a \textit{communicator} that describes the communication context and an associated group of processes, as seen in \cite{WEAUM}. But, the cost of splitting the initial communicator is non-negligible  \cite{Sack2010}. In the case of coded TeraSort of \cite{LiSMA17} the overall communicator needs to be split into ${K}\choose{r+1}$ \textit{intra-communicators} each of which facilitates the communication within a group. For this splitting to happen, all servers need to call the function MPI\_Comm\_Split. 

%
%
%

We demonstrate the impact of this issue by explicitly measuring the time needed to split the initial communicator of $K$ servers into ${K}\choose{x}$ intra-communicators, each of size $x$ for different values of $K$ and $x$. Let us refer to Fig. \ref{fig:comm_split}.
We see that the cost of MPI\_Comm\_Split incurs an exponential cost that can easily dominate the overall MapReduce execution. The instance type used for these servers is $\texttt{m3.large}$. This clearly indicates that even though the communication load may reduce with increasing $r$ in the scheme of \cite{LiSMA17}, the overall execution time may be adversely affected.

\begin{figure}[t]
\centering
\includegraphics[scale=0.6]{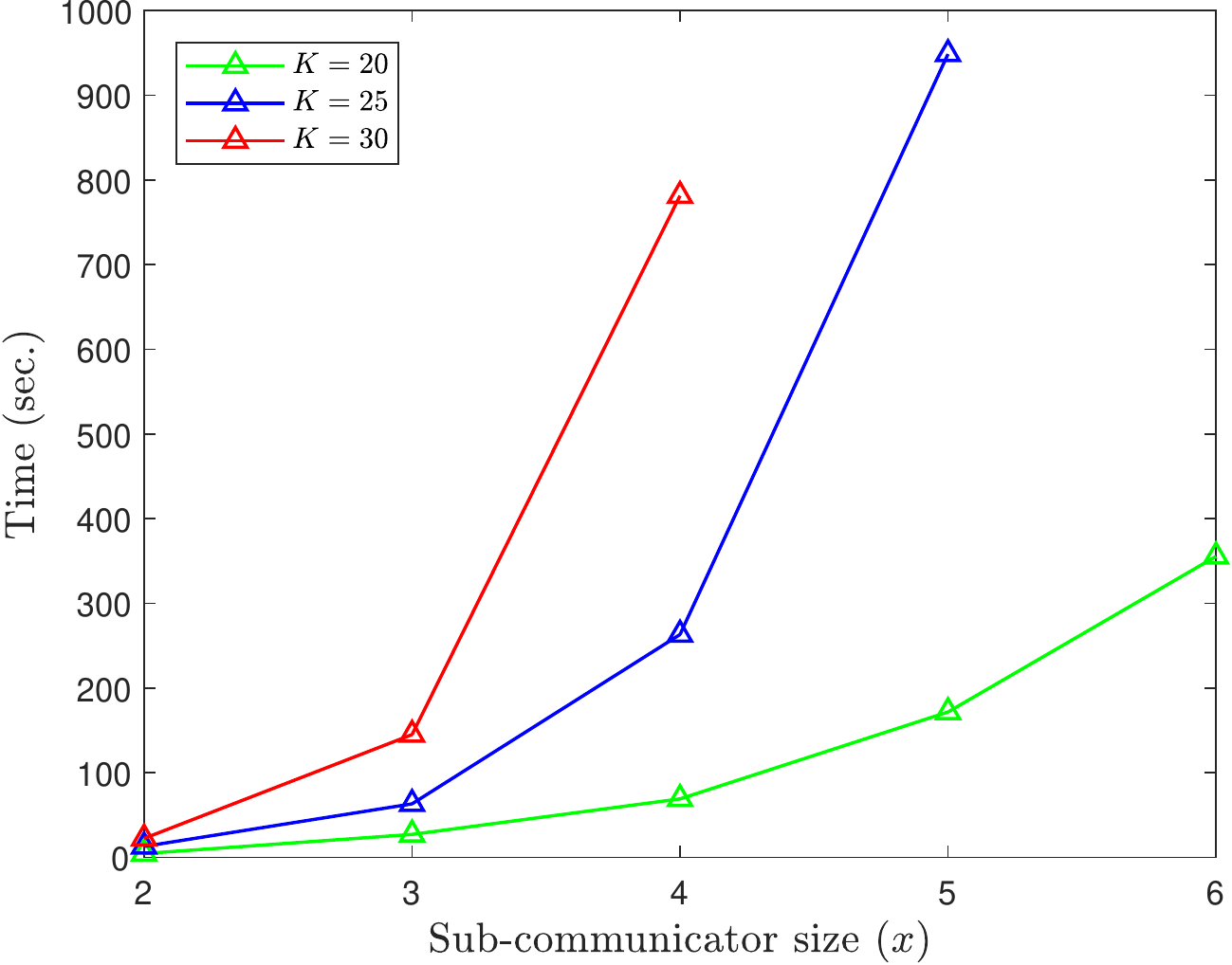}
\caption{Conventional \texttt{MPI\_Comm\_Split} execution time for splitting $K$ nodes into ${K}\choose{x}$ sub-communicators each of size $x$.}
\label{fig:comm_split}
\vspace{-0.2in}
\end{figure}

Another point to consider is that the MPI library might support a limited number of communicators. Some indicative examples are those of Open MPI \cite{OMPI} which supports up to $2^{30}-1$ communicators, MPI over InfiniBand, Omni-Path, Ethernet/iWARP, and RoCE (MVAPICH) \cite{MVAPICH} which allows for up to 2000 communicators and High-Performance Portable MPI (MPICH) \cite{MPICH} that limits this number to 16000. Thus, if we have ($K=50$, $r=10$) the number of required groups will be ${50}\choose{11}$ which would exceed these limits.

Our experiments indicate that the time consumed in memory allocation can be non-negligible. We emphasize though that our gains over prior methods hold even if we do not take the memory allocation time into account.
Nevertheless, this is an important practical issue. It may be possible to improve this via more optimized C++ code.

Another interesting aspect of our experiments is that the observed transmission rate appears to change based on the value of $r$. In our experiments we capped the transmission rate at $100$ Mbps. However, the observed rate can be as low as $61.04$ Mbps. As our experiments run on Amazon EC2, we do not have a clear explanation on the underlying reasons. Nevertheless, we point out the rates for our proposed $r=8$ and the prior scheme $r=5$ are quite close and even here, we have a large speedup.

\section{Conclusions}
\label{sec:conclusions}
In this work we presented a distributed computing protocol by leveraging the properties of resolvable designs. These designs can in turn be generated from single parity-check codes. Our protocol is in essence, a mechanism for exploring the computation vs. communication tradeoff that is inherent within the MapReduce framework. While prior work has identified and proposed techniques for exploring this tradeoff, it suffers from performance issues relating to the large number of parts that a given job needs to be partitioned into. Our proposed approach works with a significantly smaller number of parts. Experimental results indicate that our protocol results in significantly lower execution times than the prior state of the art.

\appendix

\subsection*{Proof of Lemma \ref{lemma:cons_yields_resolv_design}}
For a given $i$, we need to show that $|B_{i,l}| = q^{k-2}$ for all $0 \leq l \leq q-1$ and that $\cup_{l=0}^{q-1} B_{i,l} = [q^{k-1}]$.
Let $\Theta = [\Theta_1 ~\Theta_2~ \dots~\Theta_k]= \bfu \bfG_{SPC}$. Then,
\begin{align*}
\Theta_i = \begin{cases}
\bfu_i & i = 1, \dots, k-1,\\
\sum_{j=1}^{k-1} \bfu_j & i=k.
\end{cases}
\end{align*}
Thus, for $1 \leq i \leq k-1$, we have that $|B_{i,l}| = |\{\bfu : \bfu_i = l\}|$. This equals $q^{k-2}$ as it is the subset of all $\bfu$ with the $i$-th coordinate equal to $l$. Furthermore, it is evident that $\calP_i = \{B_{i,l}: 0 \leq l \leq q-1\}$ forms a parallel class as the $i$-th coordinate has to belong to $\{0, \dots, q-1\}$.

Next, consider the result when $i = k$. We have
\begin{align*}
\sum_{j=1}^{k-2} \bfu_j = l - \bfu_{k-1}
\end{align*}
where $l$ is fixed.
For arbitrary $\bfu_j, 1 \leq j \leq k-2$, this equation has a unique solution for $\bfu_{k-1}$.
This implies that for any $l$, $|B_{k,l}| = q^{k-2}$ and that $\calP_{k}$ forms a parallel class.

\begin{lemma}
\label{claim:intersect}
Consider a resolvable design $(X,\mathcal{A})$ constructed by the proposed construction with parameters $k$ and $q$ and parallel classes $\mathcal{P}_1,\dots,\mathcal{P}_k$. If we pick $k-1$ blocks $B_{i_1,l_1},\dots,B_{i_{k-1},l_{k-1}}$ (where $i_j\in [k]$, $l_j\in\{0,\dots,q-1\}$) from distinct parallel classes $\mathcal{P}_{i_1},\dots,\mathcal{P}_{i_{k-1}}$, then $|\cap_{j=1}^{k-1}B_{i_j,l_j}|=1$.
\end{lemma}


\begin{proof}
From the construction in Section \ref{sec:scheme}, we know that a block $B_{i,l} \in \calP_i$ is specified by
$$
B_{i,l} = \{j : \bfT_{i,j} = l\}
$$
Now, consider $B_{i_1, l_1}, \dots, B_{i_{k-1}, l_{k-1}}$ (where $i_j \in [k], l_j \in \{0, \dots, q-1\}$) that are picked from $k-1$ distinct parallel classes $\calP_{i_1}, \dots, \calP_{i_{k-1}}$. We assume that $i_1 < i_2 < \dots < i_{k-1}$. Let $\calI =  \{i_1, \dots, i_{k-1}\}$ and $\bfT_{\calI}$ denote the submatrix of $\bfT$ that corresponds to the rows indexed by $\calI$. In what follows, we show that the vector $[l_1~l_2~\dots~l_{k-1}]^T$ is a column in $\bfT_{\calI}$.

To see this first we consider the case that $\calI = \{1, \dots, k-1\}$. In this case, the message vector $\bfu = [l_1~l_2~\dots~l_{k-1}]$ is such that $[\bfu \bfG_{SPC}]_{\calI} = [l_1~l_2~\dots~l_{k-1}]^T$ so that $[l_1~l_2~\dots~l_{k-1}]^T$ is a column in $\bfT_{\calI}$. If $k \in \calI$, then we have $i_{k-1} = k$. It follows that the system of $k-1$ equations in variables $u_1, \dots, u_{k-1}$.
\begin{align*}
u_{i_1} &= l_1,\\
u_{i_2} &= l_2,\\
\mathrel{\makebox[\widthof{=}]{\vdots}}\\
u_{i_{k-2}} &= l_{k-2},\\
u_1 + u_2 + \dots + u_{k-1} &= l_{k-1},
\end{align*}
is such that it has $k-1$ variables and a unique solution over $\mathbb Z_q$. This gives us the required result.
\end{proof}

\bibliographystyle{IEEEtran}
\input{bare_conf.bbl}

\end{document}